\documentclass[12pt]{article}
\usepackage{amsmath}
\usepackage{amssymb}
\usepackage{amsthm}
\usepackage{graphicx}
\usepackage{lineno}
\usepackage{float}
\usepackage[all]{xy}
\usepackage{mathrsfs}
\usepackage{braket}

\usepackage[utf8]{inputenc}
\usepackage{url}

\usepackage{xcolor}
\usepackage{tikz}
\usetikzlibrary{tqft}

\newtheorem{theorem}{Theorem}

\begin{document}

\title{\bf Operational protocols cannot certify classicality}

\author{{Chris Fields$^1$, James F. Glazebrook$^2$, Antonino Marcian\`{o}$^3$ and Emanuele Zappala$^4$}\\ \\
{\it$^1$ Allen Discovery Center}\\
{\it Tufts University, Medford, MA 02155 USA}\\
{fieldsres@gmail.com} \\
{ORCID: 0000-0002-4812-0744} \\ \\
{\it$^2$ Department of Mathematics and Computer Science,} \\
{\it Eastern Illinois University, Charleston, IL 61920 USA} \\
{\it and} \\
{\it Adjunct Faculty, Department of Mathematics,}\\
{\it University of Illinois at Urbana-Champaign, Urbana, IL 61801 USA}\\
{jfglazebrook@eiu.edu}\\
{ORCID: 0000-0001-8335-221X}\\ \\
{\it$^3$ Center for Field Theory and Particle Physics \& Department of Physics} \\
{\it Fudan University, Shanghai, CHINA} \\
{\it and} \\
{\it Laboratori Nazionali di Frascati INFN, Frascati (Rome), ITALY} \\
{marciano@fudan.edu.cn} \\
{ORCID: 0000-0003-4719-110X} \\ \\
{\it$^4$ Department of Mathematics and Statistics,} \\
{\it Idaho State University, Pocatello, ID 83209, USA} \\
{emanuelezappala@isu.edu} \\
{ORCID: 0000-0002-9684-9441} \\ \\
}
\maketitle

\begin{abstract}
\noindent
The existence and practical utility of operational protocols that certify entanglement raises the question of whether operational protocols exist that certify the {\em absence} of entanglement, i.e. that certify separability.  We show, within a purely topological, interpretation-independent representation, that such protocols do not exist.  Classicality is therefore, as Bohr suggested, purely a pragmatic notion. 
\end{abstract}

\textbf{Keywords:} Complexity, Computation, Decoherence, Device independence, LOCC, Separability

\section{Introduction}

Quantum communication security has motivated the development of robust operational, i.e. device-independent, protocols for certifying and quantifying shared entanglement; see e.g. \cite{friis:19} for a general review and \cite{vidick:19} for the particular case of quantum-key distribution.  Such protocols involve independent observers (Alice and Bob) who each perform finite local preparation and measurement operations on shared quantum states, and who are also allowed to communicate classically before preparations and after measurements; hence they are described as (finite) {\em local operations, classical communication} (LOCC) protocols \cite{chit:14}, with a Clauser-Horne-Shimony-Holt (CHSH) test \cite{chsh:69} as a canonical example.\footnote{We will use `LOCC' to refer exclusively to finite-round LOCC protocols.}

The existence of such protocols naturally raises two questions:

\begin{enumerate}
\item Do LOCC protocols exist that certify the {\em absence} of shared entanglement, i.e. that certify a purely-classical interaction?
\item Can the parties in a LOCC protocol demonstrate by local operations that the protocol is LOCC?
\end{enumerate}

Here we show that the answers to both of these questions are `No'.  Classicality, whether in the form of state separability, decoherence, or classical communication, cannot be certified by device-independent protocols.  Hence while various forms of classicality are often assumed as information-processing resources, e.g. in defining LOCC itself, their availability as resources cannot be operationally demonstrated.

The current result generalizes previous results for particular cases.  Brand\~{a}o, Christandl and Yard \cite{bcy:11} showed that a classical verifier (Arthur) of proofs generated by $k \geq 2$ resource-unlimited (Merlin) provers could not determine that the $k$ proofs were unentangled if restricted to LOCC, i.e. that {\bf QMA}$(k) =$ {\bf QMA} under LOCC, where {\bf QMA} is the {\it Quantum Merlin-Arthur} complexity class consisting of all languages for which, given a string belonging to the language, there is a polynomial-size quantum prover that can certify the string's belonging to the class with arbitrarily high precision.  We showed in \cite{fgmz:25b} that a referee who communicates classically with multiple interactive provers that share an entangled state (a {\bf MIP*} machine) cannot demonstrate independence of the provers, so cannot distinguish between  {\bf MIP*} and monolithic quantum computation, e.g. by a quantum Turing Machine (QTM).  Local operations by Alice and Bob require that their respective operators $Q_A$ and $Q_B$ are separable, i.e. that the joint operator $Q$ factorizes as $Q = Q_A \otimes Q_B$.  Cohen \cite{cohen:15} demonstrated probabilistic state-discrimination problems solvable by separable operators unrestricted by locality but not solvable when restricted to LOCC.  It is known, moreover, that even without LOCC or other restrictions, the problem of identifying separable states in an arbitrary finite-dimensional Hilbert space is {\bf NP}-hard \cite{chit:19}.

Interpretations of quantum theory differ on whether classicality is considered objective --- see, e.g. \cite{landsman:06, schloss:15} for discussions of Bohr's position on the matter --- and modifications of quantum theory that introduce physical ``collapse'' mechanisms \cite{bassi:13} do so to guarantee objective classicality 
--- in the language of Bohr, the ability to communicate experimental setups and results in classical terms, not necessarily as the existence of observer-independent properties of quantum systems.  The methods employed here make no interpretation-specific assumptions, and the results described are interpretation independent.  In particular, they show that classicality cannot be operationally certified even if it is ``objective'' from a global, non-operational perspective.

\section{Operational locality in LOCC protocols}

Before constructing a two-observer, or two-agent, LOCC protocol, let us consider a single finite observer $A$ interacting with a finite environment $\bar{A}$, i.e. we have $A\bar A= U$, the effectively isolated ``universe'' of interest.  The interaction is dictated by the operator   $H_{A \bar{A}} = H_U - (H_A + H_{\bar{A}})$, and can be represented as defined at a boundary $\mathscr{B}_A$ separating $A$ from $\bar{A}$.  {\em Separation} is here used in its technical sense of factorization; observers/agents in LOCC protocols are distinct systems with distinct states, capable of storing classical memories and employing classical communication; therefore the joint state $\rho_{A \bar{A}} = \rho_A \rho_{\bar{A}}$ at all times.  The boundary $\mathscr{B}_A$ functions as a holographic screen \cite{bousso:02} mediating information flow between $A$ and $\bar{A}$ and can be represented, without loss of generality, as an ancillary array of $N$ qubits, where $N = \mathrm{log_2(dim(}H_{A \bar{A}}))$, and hence as an ancillary qubit space $q^N$; see \cite{fgm:22b, fg:25} for details.\footnote{Note that the Dirac-von Neumann or ``textbook'' idea that measurement involves temporary observer-environment entanglement that terminates either at the observer's discretion or by some ``collapse'' mechanism is inconsistent with this separability requirement.  Preservation of the observer-environment distinction, i.e. separability, requires that observer and environment alternately prepare and measure states of an intervening holographic screen.}  The interaction $H_{A \bar{A}}$ between $A$ and $A$'s environment, i.e. $\bar{A}$, is by definition {\em local} to $\mathscr{B}_A$.  A LOCC observer/agent $A$ acts {\em at} her boundary, and in the formalism of \cite{fgm:22b, fg:25} acts {\em on} her boundary; under no circumstances does she act {\em beyond} her boundary, though the consequences of her actions may well propagate beyond her boundary.\footnote{Referring to $A$ as an ``agent'' implies nothing about consciousness or being ``human-like'' but does imply freedom from local determinism via the Conway-Kochen theorem \cite{conway-kochen:09}.  As discussed in depth in \cite{fg:25, ieee1}, the above characterization of locally-acting observers/agents is a generalization to generic quantum systems of classical {\em active inference agents}, time-persistent components of classical dynamical systems that comply with the {\em free energy principle} of \cite{ramstead:22,friston:23}.  The boundary $\mathscr{B}_A$ in the classical theory is a Markov blanket \cite{pearl:88}.}  

Following \cite{chit:14}, the action of $A$ on $\bar{A}$ can be represented as a {\it quantum instrument}, defined for a Hilbert space $\mathcal{H}$ as a family of completely positive maps $\mathcal E_j : B(\mathcal{H}) \longrightarrow B(\mathcal{H})$, where each $\mathcal E_j$ is a bounded linear map on the space of bounded linear operators $B(\mathcal{H})$ of $\mathcal{H}$. Additionally, it is required that the sum $\sum_j \mathcal E_j$ is trace preserving, and there are countably many (we will assume only finitely many) indices $j$ --- for a review of the subject including further properties, see e.g. \cite{gudder:23}.  Associated with the instrument $\{ E_j \}$ is a quantum-classical map, defined as a trace-preserving completely positive (TCP) map $B(\mathcal{H}) \longrightarrow B(\mathcal{H})\otimes B(\mathbb C)$ of the form $\rho \mapsto \sum_j \mathcal E_j(\rho) \otimes \ket{j} \bra{j}$.  We can see from the above that the space $\mathcal{H}$ on which these operators are defined is the Hilbert space of the boundary, $\mathcal{H}_{\mathscr{B}_A}$, the qubit space $q^N$.  The values $E_j(\rho)$ are defined with respect to some basis, which we can treat as fixed by some quantum reference frame (QRF) \cite{bartlett:07}.  These values are written to, or read from, a physically-implemented data structure that serves, for $A$, as a classical memory.  As shown in \cite{fgm:22a}, sequences of operations with any fixed QRF, including memory reads and writes, can be represented by a topological quantum field theory (TQFT) \cite{atiyah:88}.  This TQFT, which we will denote $Q_A$, is implemented by the bulk internal interaction $H_A$ of $A$.  Hence we can represent $A$ as in Diagram \eqref{alice}:

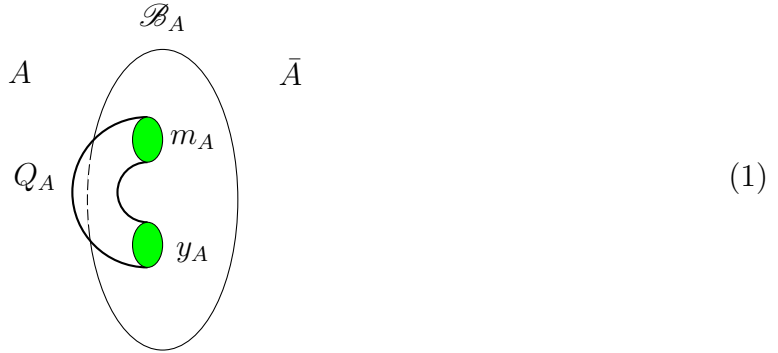
\begin{equation} \label{alice}
\begin{gathered}
\begin{tikzpicture}[every tqft/.append style={transform shape}]
\draw[rotate=90] (0,0) ellipse (2cm and 1 cm);
\node[above] at (0,2.1) {$\mathscr{B}_{A}$};
\node at (-1.7,0.3) {$Q_A$};
\node at (-1.9,1.7) {$A$};
\node at (1.7,1.7) {$\bar{A}$};
\node at (0.4,0.8) {$m_A$};
\node at (0.4,-0.7) {$y_A$};
\draw [thick] (-0.2,1.1) arc [radius=1, start angle=90, end angle= 270];
\draw [thick] (-0.2,0.5) arc [radius=0.4, start angle=90, end angle= 270];
\draw[rotate=90,fill=green,fill opacity=1] (0.8,0.2) ellipse (0.3 cm and 0.2 cm);
\draw[rotate=90,fill=green,fill opacity=1] (-0.6,0.2) ellipse (0.3 cm and 0.2 cm);
\draw [ultra thick, white] (-1,0.5) -- (-0.8,0.5);
\draw [ultra thick, white] (-1,0.3) -- (-0.8,0.3);
\draw [ultra thick, white] (-1.1,0.1) -- (-0.8,0.1);
\draw [ultra thick, white] (-1.1,-0.1) -- (-0.8,-0.1);
\draw [ultra thick, white] (-1,-0.3) -- (-0.8,-0.3);
\end{tikzpicture}
\end{gathered}
\end{equation}

Here the green ellipses labeled `$m_A$' and `$y_A$' represent the sectors of $\mathscr{B}_A$, each comprising some proper subset of qubits of $\mathscr{B}_A$, on which $Q_A$ acts.  

A two-agent LOCC protocol is an interaction between two systems of the type illustrated in Diagram \eqref{alice}, $A$ and $B$, communicating via distinct quantum and classical channels traversing their shared environment $E$ \cite{fgm:24a}:

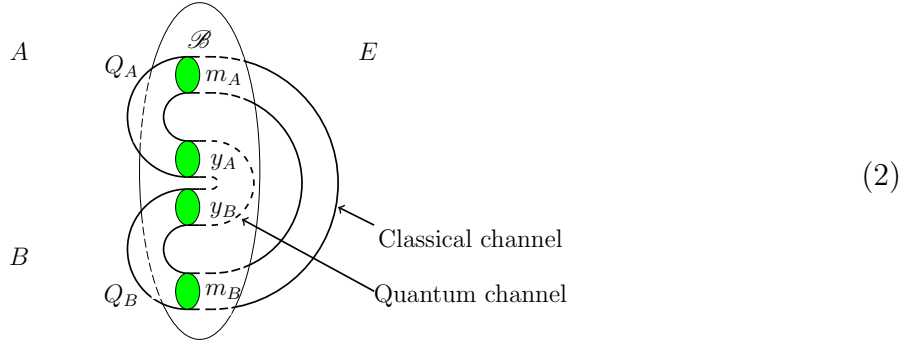
\begin{equation} \label{locc-diag}
\begin{gathered}
\scalebox{0.8}{
\begin{tikzpicture}[every tqft/.append style={transform shape}]
\draw[rotate=90] (0,0) ellipse (2.8cm and 1 cm);
\node[above] at (0,1.9) {$\mathscr{B}$};
\draw [thick] (-0.2,1.9) arc [radius=1, start angle=90, end angle= 270];
\draw [thick] (-0.2,1.3) arc [radius=0.4, start angle=90, end angle= 270];
\draw[rotate=90,fill=green,fill opacity=1] (1.6,0.2) ellipse (0.3 cm and 0.2 cm);
\draw[rotate=90,fill=green,fill opacity=1] (0.2,0.2) ellipse (0.3 cm and 0.2 cm);
\draw [thick] (-0.2,-0.3) arc [radius=1, start angle=90, end angle= 270];
\draw [thick] (-0.2,-0.9) arc [radius=0.4, start angle=90, end angle= 270];
\draw[rotate=90,fill=green,fill opacity=1] (-0.6,0.2) ellipse (0.3 cm and 0.2 cm);
\draw[rotate=90,fill=green,fill opacity=1] (-2.0,0.2) ellipse (0.3 cm and 0.2 cm);
\draw [rotate=180, thick, dashed] (-0.2,0.9) arc [radius=0.7, start angle=90, end angle= 270];
\draw [rotate=180, thick, dashed] (-0.2,0.3) arc [radius=0.1, start angle=90, end angle= 270];
\draw [thick] (-0.2,0.5) -- (0,0.5);
\draw [thick] (-0.2,-0.1) -- (0,-0.1);
\draw [thick] (-0.2,-0.9) -- (0,-0.9);
\draw [thick] (-0.2,-0.3) -- (0,-0.3);
\draw [thick, dashed] (0,0.5) -- (0.2,0.5);
\draw [thick, dashed] (0,-0.1) -- (0.2,-0.1);
\draw [thick, dashed] (0,-0.9) -- (0.2,-0.9);
\draw [thick, dashed] (0,-0.3) -- (0.2,-0.3);
\node[above] at (-3,1.7) {$A$};
\node[above] at (-3,-1.7) {$B$};
\node[above] at (2.8,1.7) {$E$};
\draw [ultra thick, white] (-0.9,1.5) -- (-0.7,1.5);
\draw [ultra thick, white] (-1,1.3) -- (-0.8,1.3);
\draw [ultra thick, white] (-1,1.1) -- (-0.8,1.1);
\draw [ultra thick, white] (-1,0.9) -- (-0.8,0.9);
\draw [ultra thick, white] (-1.1,0.7) -- (-0.8,0.7);
\draw [ultra thick, white] (-1.1,0.5) -- (-0.8,0.5);
\draw [ultra thick, white] (-1,-0.9) -- (-0.8,-0.9);
\draw [ultra thick, white] (-1,-1.1) -- (-0.8,-1.1);
\draw [ultra thick, white] (-1,-1.3) -- (-0.8,-1.3);
\draw [ultra thick, white] (-0.9,-1.5) -- (-0.7,-1.5);
\draw [ultra thick, white] (-0.9,-1.7) -- (-0.7,-1.7);
\draw [ultra thick, white] (-0.8,-1.9) -- (-0.6,-1.9);
\draw [ultra thick, white] (-0.8,-2.1) -- (-0.6,-2.1);
\node[above] at (-1.3,1.4) {$Q_A$};
\node[above] at (-1.3,-2.4) {$Q_B$};
\draw [rotate=180, thick] (-0.2,2.3) arc [radius=2.1, start angle=90, end angle= 270];
\draw [rotate=180, thick] (-0.2,1.7) arc [radius=1.5, start angle=90, end angle= 270];
\draw [thick] (-0.2,1.9) -- (0,1.9);
\draw [thick] (-0.2,1.3) -- (0,1.3);
\draw [thick, dashed] (0.2,1.9) -- (0,1.9);
\draw [thick, dashed] (0.2,1.3) -- (0,1.3);
\draw [thick] (-0.2,-1.7) -- (0,-1.7);
\draw [thick] (-0.2,-2.3) -- (0,-2.3);
\draw [thick, dashed] (0.2,-1.7) -- (0,-1.7);
\draw [thick, dashed] (0.2,-2.3) -- (0,-2.3);
\draw [ultra thick, white] (0.3,2) -- (0.3,1.2);
\draw [ultra thick, white] (0.5,2) -- (0.5,1.2);
\draw [ultra thick, white] (0.7,1.9) -- (0.7,1.1);
\draw [ultra thick, white] (0.3,-2.4) -- (0.3,-1.5);
\draw [ultra thick, white] (0.5,-2.4) -- (0.5,-1.5);
\draw [ultra thick, white] (0.7,-1.8) -- (0.7,-1.5);
\node at (0.4,1.55) {$m_A$};
\node at (0.4,0.15) {$y_A$};
\node at (0.4,-0.65) {$y_B$};
\node at (0.4,-2) {$m_B$};
\node[above] at (4.5,-2.4) {Quantum channel};
\draw [thick, ->] (2.9,-2) -- (0.7,-0.8);
\node[above] at (4.5,-1.4) {Classical channel};
\draw [thick, ->] (2.9,-0.9) -- (2.3,-0.6);
\end{tikzpicture}
}
\end{gathered}
\end{equation}
Here $A$ and $B$ are restricted to interacting only via the designated channels traversing $E$ by requiring that the direct interaction $H_{AB} = 0$; in a practical setting, this is accomplished by imposing a spacelike separation.  Note both that $A$ and $B$ being mutually separable is required for the assumption of classical communication via a causal channel in $E$, and that this assumption renders $Q_A$ and $Q_B$ noncommutative, and hence subject to quantum contextuality \cite{fg:21,fg:23} --- see also \cite{adlam:21} for a comprehensive review including related concepts.  We showed in \cite{fgm:24a} that both canonical Bell/EPR experiments and multi-observer quantum Darwinism \cite{zurek:09} can be represented in the form of Diagram \eqref{locc-diag}.  One-way classical communication, e.g. between measurements of initial and final states in a scattering experiment, can be imposed is by ``unfolding'' Diagram \eqref{locc-diag} and imposing an external (i.e. measured by $E$) flow of time:

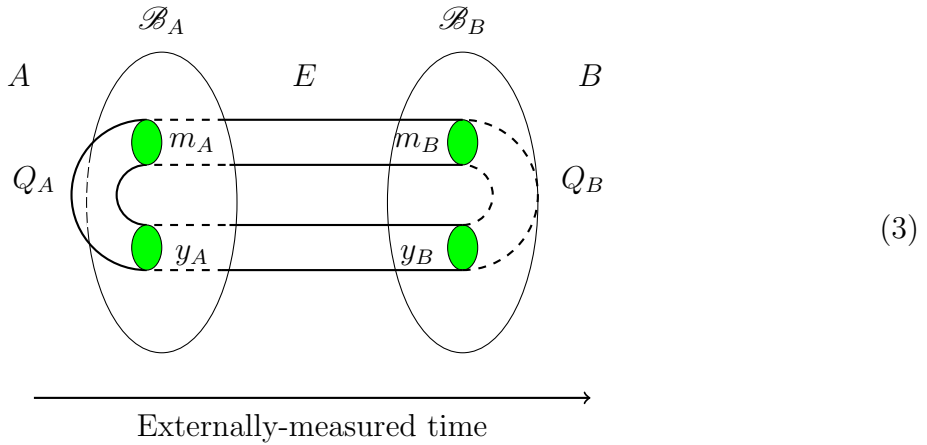
\begin{equation} \label{causal-locc}
\begin{gathered}
\begin{tikzpicture}[every tqft/.append style={transform shape}]
\draw[rotate=90] (0,0) ellipse (2cm and 1 cm);
\node[above] at (0,2.1) {$\mathscr{B}_{A}$};
\draw [thick] (-0.2,1.1) arc [radius=1, start angle=90, end angle= 270];
\draw [thick] (-0.2,0.5) arc [radius=0.4, start angle=90, end angle= 270];
\draw[rotate=90,fill=green,fill opacity=1] (0.8,0.2) ellipse (0.3 cm and 0.2 cm);
\draw[rotate=90,fill=green,fill opacity=1] (-0.6,0.2) ellipse (0.3 cm and 0.2 cm);
\draw [ultra thick, white] (-1,0.5) -- (-0.8,0.5);
\draw [ultra thick, white] (-1,0.3) -- (-0.8,0.3);
\draw [ultra thick, white] (-1.1,0.1) -- (-0.8,0.1);
\draw [ultra thick, white] (-1.1,-0.1) -- (-0.8,-0.1);
\draw [ultra thick, white] (-1,-0.3) -- (-0.8,-0.3);
\draw[rotate=90] (0,-4) ellipse (2cm and 1 cm);
\draw[rotate=90,fill=green,fill opacity=1] (0.8,-4) ellipse (0.3 cm and 0.2 cm);
\draw[rotate=90,fill=green,fill opacity=1] (-0.6,-4) ellipse (0.3 cm and 0.2 cm);
\draw [thick, dashed] (-0.1,0.5) -- (0.9,0.5);
\draw [thick, dashed] (-0.1,1.1) -- (0.8,1.1);
\draw [thick, dashed] (-0.1,-0.3) -- (0.9,-0.3);
\draw [thick, dashed] (-0.1,-0.9) -- (0.8,-0.9);
\draw [thick] (0.9,0.5) -- (4,0.5);
\draw [thick] (0.8,1.1) -- (4,1.1);
\draw [thick] (0.9,-0.3) -- (4,-0.3);
\draw [thick] (0.8,-0.9) -- (4,-0.9);
\draw [rotate=180, thick, dashed] (-4,0.9) arc [radius=1, start angle=90, end angle= 270];
\draw [rotate=180, thick, dashed] (-4,0.3) arc [radius=0.4, start angle=90, end angle= 270];
\node[above] at (4,2.1) {$\mathscr{B}_{B}$};
\node at (-1.7,0.3) {$Q_A$};
\node at (5.6,0.3) {$Q_B$};
\node at (-1.9,1.7) {$A$};
\node at (5.7,1.7) {$B$};
\node at (1.9,1.7) {$E$};
\node at (0.4,0.8) {$m_A$};
\node at (0.4,-0.7) {$y_A$};
\node at (3.4,0.8) {$m_B$};
\node at (3.4,-0.7) {$y_B$};
\draw [thick, ->] (-1.7,-2.6) -- (5.7,-2.6);
\node at (2,-3) {Externally-measured time};
\end{tikzpicture}
\end{gathered}
\end{equation}

Here it is also natural to think of $B$ as a ``future version'' of $A$, in which case the classical channel is $A$'s classical memory.

It is worth emphasizing that Diagrams \eqref{alice}, \eqref{locc-diag}, and \eqref{causal-locc} do not assume a background space, and hence define ``locality'' purely topologically.  State preparation and measurement actions performed by $A$ and $B$ using $Q_A$ and $Q_B$, respectively, are {\em local} to the boundaries $\mathscr{B}_{A}$ and $\mathscr{B}_{B}$, respectively, and indeed local to the sectors $(m_A, y_A)$ and $(m_B, y_B)$ of those boundaries on which the operations $Q_A$ and $Q_B$ are respectively defined.  The quantum and classical channels connecting these boundary sectors are implemented by $E$, and are strictly external, and therefore not local, to both $A$ and $B$.  This topological notion of locality follows from the assumption that $A$ and $B$ are distinct from both each other and from $E$ --- hence the separability condition $\rho_{ABE} = \rho_A \rho_B \rho_E$ is satisfied --- together with the Hilbert-space formalism, and requires no interpretation-dependent assumptions.  Topological locality therefore underlies {\em all} physically-realizable protocols involving preparation and measurement operations performed by systems/agents that remain distinct from both each other and their shared environment, and does so in an interpretation-independent way.

\section{Operational locality constrains measurements of entanglement entropy}

The fundamental purpose of LOCC protocols is to measure the entanglement entropy between the locally-accessible ``ends'' of the quantum channel, i.e. the entanglement entropy $\mathcal{S}(\rho_{q_A q_B})$ of the joint state $\rho_{q_A q_B}$ of the boundary sectors $q_A$ and $q_B$.  In protocols modeled on Bell/EPR experiments, this is generally accomplished by measuring the CHSH expectation value:

\begin{equation} \label{chsh-def}
\begin{split}
EXP = | & <<A_1,B_1>> + <<A_1,B_2>> + \\
& <<A_2,B_1>> - <<A_2,B_2>>|,
\end{split}
\end{equation}
\noindent
where $<<x,y>>$ denotes the expectation value for a collection of joint measurements of $x$ and $y$.  If $EXP > 2$, classical data reported by $A$ and $B$ violate the CHSH inequality, indicating entanglement between $q_A$ and $q_B$ \cite{chsh:69}; if  $q_A$ and $q_B$ each comprise single qubits, the upper limit is $EXP \leq 2 \surd 2$, the relevant Tsirelson bound \cite{cirelson:80}.  In protocols based on quantum Darwinism, violation of the CHSH inequality is assumed when $A$ and $B$ report observations of the {\em same} joint state $\rho_{q_A q_B}$ and hence observations of the same quantum channel, i.e. the same ``external system'' embedded in $E$ \cite{fgm:24a}.

Measurements of $\mathcal{S}(\rho_{q_A q_B})$ provide indirect measures of entanglement entropy between the quantum channel and the rest of $E$; the joint state $\rho_{q_A q_B}$ being measurably entangled indicates protection of the quantum channel from decohering interactions with the rest of $E$.  Monogamous entanglement of $\rho_{q_A q_B}$ indicates zero decoherence, and hence no interaction between the quantum channel and rest of $E$; as the latter condition is consistent with topological identification of $q_A$ with $q_B$, the Maldecena-Susskind ``ER = EPR'' conjecture \cite{maldecena:13} is an operational theorem for LOCC protocols \cite{fgmz:24}.

The operational locality enforced by LOCC disallows, however, measurements of any other entanglement entropies.  The relevant results can be stated as theorems:

\begin{theorem}[\cite{fg:23}, Thm. 3.1] \label{thm1}
No finite system $A$ can measure the entanglement entropy $\mathcal{S}(\bar{A})$ of its environment $\bar{A}$.
\end{theorem}

\begin{theorem}[\cite{fg:23}, Cor. 3.1] \label{thm2}
No finite system $A$ can measure the entanglement entropy $\mathcal{S}(\rho_{A \bar{A}})$ across the boundary $\mathscr{B}_A$ that separates it from its environment $\bar{A}$.
\end{theorem}

\begin{theorem}[\cite{fgmz:25b}, Thm. 2] \label{thm3}
An observer $C$ embedded in the shared environment $E$ of systems $A$ and $B$ cannot determine, either by monitoring classical communication between $A$ and $B$ or by performing local measurements within $E$, whether $A$ and $B$ are employing a LOCC protocol with quantum and classical channels traversing $E$.
\end{theorem}

The proofs can be sketched as follows:

\begin{proof} [Thm. 1 (sketch)]
Separability between $A$ and $\bar{A}$ requires that the interaction $H_{A \bar{A}}$ is weak compared to the bulk self-interactions $H_A$ and $H_{\bar{A}}$.  Equivalently, separability requires that the bulk Hilbert-space dimensions ${\rm dim}(\mathcal{H}_{A}), {\rm dim}(\mathcal{H}_{\bar{A}}) \gg {\rm dim}(H_{A \bar{A}}) = {\rm dim}(\mathscr{B}_A) = 2^N $.  The state $\rho_{\bar{A}}$ cannot, therefore, be determined by measuring $\rho_{\mathscr{B}_A}$ --- indeed, $A$ also cannot measure the full states $\rho_{\mathscr{B}_A}$ or $\rho_A$ for reasons detailed in \cite{fg:25}.  Any measurement of an internal state of $\bar{A}$ is nonlocal for $A$, and hence forbidden by LOCC. 
\end{proof}

\begin{proof} [Thm. 2 (sketch)]
The interaction $H_{A \bar{A}}$ is independent of tensor-product decompositions of $\bar{A}$, and as noted above, $A$ cannot measure the total state $\rho_{\bar{A}}$.
\end{proof}

\begin{proof} [Thm. 3 (sketch)]

The observer $C$, Charlie, is situated as in:

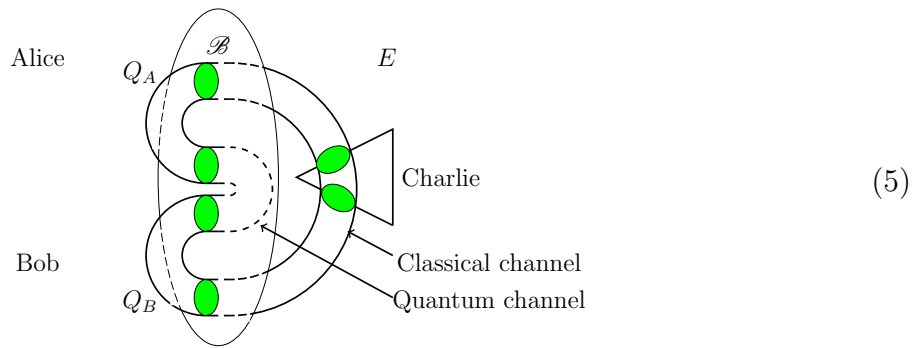
\begin{equation} \label{locc-C-diag}
\begin{gathered}
\scalebox{0.8}{
\begin{tikzpicture}[every tqft/.append style={transform shape}]
\draw[rotate=90] (0,0) ellipse (2.8cm and 1 cm);
\node[above] at (0,1.9) {$\mathscr{B}$};
\draw [thick] (-0.2,1.9) arc [radius=1, start angle=90, end angle= 270];
\draw [thick] (-0.2,1.3) arc [radius=0.4, start angle=90, end angle= 270];
\draw[rotate=90,fill=green,fill opacity=1] (1.6,0.2) ellipse (0.3 cm and 0.2 cm);
\draw[rotate=90,fill=green,fill opacity=1] (0.2,0.2) ellipse (0.3 cm and 0.2 cm);
\draw [thick] (-0.2,-0.3) arc [radius=1, start angle=90, end angle= 270];
\draw [thick] (-0.2,-0.9) arc [radius=0.4, start angle=90, end angle= 270];
\draw[rotate=90,fill=green,fill opacity=1] (-0.6,0.2) ellipse (0.3 cm and 0.2 cm);
\draw[rotate=90,fill=green,fill opacity=1] (-2.0,0.2) ellipse (0.3 cm and 0.2 cm);
\draw [rotate=180, thick, dashed] (-0.2,0.9) arc [radius=0.7, start angle=90, end angle= 270];
\draw [rotate=180, thick, dashed] (-0.2,0.3) arc [radius=0.1, start angle=90, end angle= 270];
\draw [thick] (-0.2,0.5) -- (0,0.5);
\draw [thick] (-0.2,-0.1) -- (0,-0.1);
\draw [thick] (-0.2,-0.9) -- (0,-0.9);
\draw [thick] (-0.2,-0.3) -- (0,-0.3);
\draw [thick, dashed] (0,0.5) -- (0.2,0.5);
\draw [thick, dashed] (0,-0.1) -- (0.2,-0.1);
\draw [thick, dashed] (0,-0.9) -- (0.2,-0.9);
\draw [thick, dashed] (0,-0.3) -- (0.2,-0.3);
\node[above] at (-3,1.7) {Alice};
\node[above] at (-3,-1.7) {Bob};
\node[above] at (2.8,1.7) {$E$};
\draw [ultra thick, white] (-0.9,1.5) -- (-0.7,1.5);
\draw [ultra thick, white] (-1,1.3) -- (-0.8,1.3);
\draw [ultra thick, white] (-1,1.1) -- (-0.8,1.1);
\draw [ultra thick, white] (-1,0.9) -- (-0.8,0.9);
\draw [ultra thick, white] (-1.1,0.7) -- (-0.8,0.7);
\draw [ultra thick, white] (-1.1,0.5) -- (-0.8,0.5);
\draw [ultra thick, white] (-1,-0.9) -- (-0.8,-0.9);
\draw [ultra thick, white] (-1,-1.1) -- (-0.8,-1.1);
\draw [ultra thick, white] (-1,-1.3) -- (-0.8,-1.3);
\draw [ultra thick, white] (-0.9,-1.5) -- (-0.7,-1.5);
\draw [ultra thick, white] (-0.9,-1.7) -- (-0.7,-1.7);
\draw [ultra thick, white] (-0.8,-1.9) -- (-0.6,-1.9);
\draw [ultra thick, white] (-0.8,-2.1) -- (-0.6,-2.1);
\node[above] at (-1.3,1.4) {$Q_A$};
\node[above] at (-1.3,-2.4) {$Q_B$};
\draw [rotate=180, thick] (-0.2,2.3) arc [radius=2.1, start angle=90, end angle= 270];
\draw [rotate=180, thick] (-0.2,1.7) arc [radius=1.5, start angle=90, end angle= 270];
\draw [thick] (-0.2,1.9) -- (0,1.9);
\draw [thick] (-0.2,1.3) -- (0,1.3);
\draw [thick, dashed] (0.2,1.9) -- (0,1.9);
\draw [thick, dashed] (0.2,1.3) -- (0,1.3);
\draw [thick] (-0.2,-1.7) -- (0,-1.7);
\draw [thick] (-0.2,-2.3) -- (0,-2.3);
\draw [thick, dashed] (0.2,-1.7) -- (0,-1.7);
\draw [thick, dashed] (0.2,-2.3) -- (0,-2.3);
\draw [ultra thick, white] (0.3,2) -- (0.3,1.2);
\draw [ultra thick, white] (0.5,2) -- (0.5,1.2);
\draw [ultra thick, white] (0.7,1.9) -- (0.7,1.1);
\draw [ultra thick, white] (0.3,-2.4) -- (0.3,-1.5);
\draw [ultra thick, white] (0.5,-2.4) -- (0.5,-1.5);
\draw [ultra thick, white] (0.7,-1.8) -- (0.7,-1.5);
\node[above] at (4.5,-2.4) {Quantum channel};
\draw [thick, ->] (2.9,-2) -- (0.7,-0.8);
\node[above] at (4.5,-1.7) {Classical channel};
\draw [thick, ->] (2.9,-1.3) -- (2.15,-0.9);
\draw [thick] (2.9, -0.8) -- (2.9, 0.8) -- (1.3, 0.0) -- (2.9, -0.8);
\draw[rotate=30,fill=green,fill opacity=1] (1.8, -0.7) ellipse (0.3 cm and 0.2 cm);
\draw[rotate=-30,fill=green,fill opacity=1] (1.9, 0.7) ellipse (0.3 cm and 0.2 cm);
\node at (3.7, 0.0) {Charlie};
\end{tikzpicture}
}
\end{gathered}
\end{equation}

Clearly, $C$ can determine whether the observations communicated by $A$ and $B$ violate the CHSH inequality; this is exactly the role of the referee in a CHSH game, or of a self-tester in an entanglement-certification protocol \cite{friis:19}.  However, classical communications between $A$ and $B$ cannot determine $\mathcal{S}(\rho_{AB})$, by Thm. \ref{thm2} above.  Theorem \ref{thm1} above rules out $C$ measuring $\mathcal{S}(\bar{C})$, where $C\bar{C} = U$, and the decomposition-independence of $\mathcal{H}_{C \bar{C}}$ forbids any local measurement by $C$ determining the location of boundaries, including $\mathscr{B}_A$ and $\mathscr{B}_B$, within $\bar{C}$.
\end{proof}

Informally, operational locality prevents any system from demonstrating that its environment comprises mutually-separable systems (Thm. \ref{thm1}), that it itself is separate from its environment (Thm. \ref{thm2}), or that classical communications that it is monitoring are reporting operationally-local outcomes (Thm. \ref{thm3}).  A specific consequence of Thm. \ref{thm3} is that purported MIP* machines cannot be shown by topologically-local operations to be MIP*; in particular, the purported independence of the provers cannot be operationally certified \cite{fgmz:25b}.  This result generalizes, as shown in the next section.

\section{Separability cannot be operationally certified}

Theorem \ref{thm2} answers the second question raised in the Introduction: if $A$ cannot show that it is separate from $\bar{A}$, it cannot show that it is engaged in a LOCC protocol with a distinct system $B$ embedded in $\bar{A}$.  It cannot, in particular, show by local operations that $H_{AB} = 0$ as required by LOCC.  From the perspective of its participants, therefore, a protocol being LOCC is an {\em a priori} stipulation, not an observational outcome.  Theorem \ref{thm3} confirms this result from the perspective of any third party.

Without loss of generality, therefore, let us assume that a particular protocol is LOCC, and address the first question raised in the Introduction: can $A$ and $B$ certify, via the protocol, that the state $\rho_{y_A y_B} = \rho_{y_A} \rho_{y_B}$, i.e. that the state they share via their quantum channel is {\em not} entangled?  We will show in this Section that they cannot.

Compliance with the CHSH inequality suggests separability, but the existence of entangled pairs that respect the CHSH inequality, e.g. Werner states in appropriate parameter ranges \cite{werner}, shows that CHSH compliance cannot be criterial.  Without a shared, effectively classical clock, i.e. a clock separable from both $A$ and $B$, $A$ and $B$ cannot use causality as evidence of separability, and assuming a shared clock begs the question, as $A$ and $B$ must regard the signals they each receive from the clock as independent, i.e. separable, to use the clock as a time standard.  Their classical communication channel, moreover, only allows $A$ and $B$ to exchange observational outcomes. Neither has observational access to the entire joint state $\rho_{y_A y_B}$, so neither observer can determine by direct observation that it is separable.  

As an example, consider quantum state tomography.  The full reconstruction of the joint state $\rho_{y_Ay_B}$ can be in principle obtained via quantum state tomography, which can in certain practical cases be performed with high accuracy via local operations --- see e.g. \cite{xin2019local}. If $A$ wants to reconstruct $\rho_{y_Ay_B}$ using local quantum state tomography, $A$ needs to access all expectation values of operators of the form $T^A_i\otimes T^B_j$, i.e. $A$ needs to compute $\langle T^A_i\otimes T^B_j\rangle = {\rm Tr}((T^A_i\otimes T^B_j)\rho_{y_Ay_B})$, where $T^A_i\otimes T^B_j$ is a basis for the bounded linear operators on $\mathcal H_A$ and $\mathcal H_B$ respectively. While our attention is limited to the finite dimensional case, for infinite dimensional separable Hilbert spaces this can also be seen analogously by considering the strong closure of algebraic tensor product as above. 

However, in order to perform the measurements of all $\langle T^A_i\otimes T^B_j\rangle$, non-classical communication between $A$ and $B$ is needed, as the trace cannot be computed locally by $A$ and $B$ separately, and communicated between the two. In fact, what $A$ and $B$ can access locally are the marginal distributions, i.e. the partial traces $\rho_{y_A} = {\rm Tr}_B(\rho_{y_Ay_B})$ and $\rho_{y_B} = {\rm Tr}_A(\rho_{y_Ay_B})$, that are known not to be sufficient to reconstruct $\rho_{y_Ay_B}$ in general. Observe that separability is a necessary condition for being able to reconstruct $\rho_{y_Ay_B}$ from $\rho_{y_A}$ and $\rho_{y_B}$, although not sufficient. It is often assumed that $A$ and $B$ know the joint state $\rho_{y_A y_B}$, e.g. in defining the Peres-Horodecki criterion \cite{peres:96, horodecki:96}, presumably via classical communication with some third party who has access to the state and has measured it.  The assumption of classical communication clearly begs the question in the current context. The fundamental problem that $A$ and $B$ face, however, is circularity: any LOCC strategy for demonstrating that $\rho_{y_A y_B}$ is separable depends on the assumption of classical communication, i.e. on the assumption that the joint state of $A$'s and $B$'s ``ends'' of the classical channel, $\rho_{m_A m_B}$, is separable. 

The circularity of the situation becomes obvious when the question is rephrased: can $A$ and $B$ certify, via a LOCC protocol, that the state $\rho_{m_A m_B} = \rho_{m_A} \rho_{m_B}$, i.e. that their purportedly classical channel is indeed classical?  The states $\rho_{m_A}$ and $\rho_{m_B}$ are states of sectors of $\mathscr{B}_A$ and $\mathscr{B}_B$, respectively; they are physical states that can, without loss of generality, be represented by joint states of finitely many qubits \cite{fg:25}.  As emphasized by Tipler \cite{tipler:14} and others, measuring the state of a physical carrier of information is measuring a quantum state.  Measurements of the quantum-channel states $\rho_{y_A}$ and $\rho_{y_B}$ provide no information one way or the other about the separability of $\rho_{m_A m_B}$, which depends only on whether the {\em classical} channel traversing $E$ is fully decohered by its interaction with the rest of $E$.\footnote{The {\em mechanism} of decoherence, which as shown in \cite{garcia-perez:20} may not involve entanglement between channel degrees of freedom and the rest of $E$, is observationally inaccessible to both $A$ and $B$ and hence is irrelevant to what $A$ or $B$ can infer from observations of $\rho_{y_A}$ and $\rho_{y_B}$.}  By Thm. \ref{thm1} above, neither $A$ nor $B$ can determine by local operations whether this is the case.

We can summarize these results as:

\begin{theorem} \label{thm4}
Operational protocols cannot certify separability of shared states.
\end{theorem}

\begin{proof}
From Thm. \ref{thm1} above, a participant $A$ in an operational protocol cannot determine the entanglement entropy $\mathcal{S}(\bar{A})$ of its environment $\bar{A}$.  Such an $A$ cannot, therefore, certify that the joint state $\rho_q \rho_{\xi}$ of any sector $q$ of its boundary and any component $\xi$ of $\bar{A}$ is separable.  Such an $A$ cannot, therefore, certify that it is engaged in an operational protocol with another observer, or with any other classical resource.  
\end{proof}

As this result depend only on topological locality, it is unchanged by geometric considerations, e.g. by embedding Diagram \eqref{causal-locc} in a spatial background that enforces a ``large'' distance between $A$ and $B$.  Einstein's idea that ``at a specific time, [these] things claim an existence independent of one another, insofar as [these] things `lie in different parts of space''' \cite{einstein:48}, i.e. that spatial separation implies separability, underlay his discomfort with entanglement as a physical phenomenon.  Spatial separation, however, must enforce separability either in all cases or none, as it cannot be shown by local operations to enforce separability in particular cases.  As dozens of experiments demonstrating, relative to an assumption of separability of the observers, entanglement between spacelike-separated systems from \cite{aspect:82} onward indicate, spatial separation alone does not guarantee joint-state separability.  Adding a physical ``collapse'' dynamics to $E$ is similarly ineffective, as whether any particular channel state has collapsed cannot be determined by local operations.

\section{Pragmatic classicality}

Grinbaum \cite{grinbaum:17} suggested that device-independent protocols, e.g. LOCC, render the notion of ``physical systems'' superfluous, but did not extend this conclusion to the participants in or users of such protocols, or to the joint environment with which they interact, all of which are necessary for the definition of LOCC.  The present analysis shows that protocols that require topologically-local operations --- including all device-independent protocols --- cannot certify either the separability of agents or the classicality of communication on which their definitions depend.  Classicality, whether in the form of separability, decoherence, classical communication, classical memory, or conditional independence, can only, therefore, be regarded as a resource for such protocols by {\em a priori} stipulation.  As quantum theory does not require classicality in any of these forms, such stipulations are pragmatic, not principled.  It is plausible to view Bohr as arguing for this very point in his discussion of the pragmatic need for ``classical concepts'' when conceptualizing or communicating about experiments \cite{bohr:58}.  

One purpose of LOCC protocols is to provide a mechanism for intersubjective agreement among observers who share access to a quantum state; quantum Darwinism has such agreement as an explicit objective \cite{zurek:09}.  Note, however, that LOCC supports intersubjective agreement only via classical communication of measurement outcomes, as in e.g. \cite{rovelli-friends:24}; direct measurements of $A$ by $B$, as advocated e.g. in \cite{adlam-rovelli:23}, are forbidden by the requirement that $H_{AB} = 0$.  We can see from the above that intersubjective agreement can only be a pragmatic notion, a result obtained already by Quine via an analysis of language use \cite{quine:60}.

From a practical perspective, theoretical conclusions in quantum computational complexity or other areas that depend on assumptions of classicality as a resource, e.g. on assumptions of independence between provers or other systems as in \cite{ji:21}, cannot be regarded as operationally verifiable \cite{fgmz:25b}.  Even probabilistic measures are, to the extent that they rely on LOCC or any other assumptions of separability of or classical communication between observers, effectively circular.

We have, in the above, employed quantum theory without further assumptions to formalize LOCC, and hence to formalize the processes of observation and inter-observer communication, in a purely-topological, interpretation-independent way.  The topological locality of measurements is a straightforward consequence of defining an observer's interaction with its environment, e.g. $H_{A \bar{A}}$, at the decompositional boundary separating states of the observer from states of its environment, i.e. the boundary implied by the factorization $U = A \bar{A}$ when $A$ and $\bar{A}$ are required, by stipulation, to be distinct systems.  We have shown, in this minimal and interpretation-independent setting, that classicality, including the classical memory of observers that allows them to accumulate observational outcomes and compute statistical tests such as CHSH, is purely pragmatic.  The question of how physics can ``appear classical'' is, in this setting, effectively the question of how some states can function, in some contexts, as classical memories.\footnote{A shared language provides classicality, in the form of jointly recognized and named objects, as a pragmatic resource to multiple communicating agents by enabling joint minimization of a variational free energy measuring medium-term Bayesian prediction errors \cite{friston-nbr:24}.}  


\section{Conclusion}

Our point in this paper is essentially G\"{o}delian: that we apply our best theories also to ourselves, and not assume, despite our classical intuitions, that we have powers of observation, inference, or action that our best theories tell us no finite system can have.  When we do this, we find that classicality --- and with it, all aspects of ourselves and our laboratories that we comfortably treat as classical --- can only be regarded as a pragmatic stipulation.  In particular, the ubiquitous assumptions of separability and consequently, full knowledge of shared joint states, that are made in definitions of operational protocols can only be regarded as stipulations.  Such stipulations of classicality are, as Bohr emphasized, essential in practice, but this pragmatic role does not support classicality being a ``real'' physical resource. 

In this G\"{o}delian stance we are in consonance with Hawking, who in \cite{hawking-incompleteness} insisted that self-referentiality and (self) inconsistent assumptions, for instance, are elements pervading theoretical physics that account for its in-principle incompleteness.  In his magisterial review of the interpretative landscape, Landsman \cite{landsman:06b} characterized as an ``extreme position'' that ``the classical world has only `relative' or `perspectival' existence'' (p. 419).  Our results suggest that we consider ourselves, and our experimental manipulations and outcomes, from this very position.

\bibliographystyle{unsrt} 
\bibliography{references}
\end{document}